\documentclass[acmsmall,screen]{acmart}
\usepackage{braket}
\usepackage{fancyhdr}
\usepackage{booktabs}

\usepackage[T1]{fontenc}
\usepackage[utf8]{inputenc}
\usepackage{regexpatch}
\usepackage{letltxmacro}
\usepackage{aligned-overset}

\newcommand{\code}[1]{{\small\texttt{#1}}}

\usepackage{outlines}

\usepackage{alltt}

\usepackage{listings,lstcoq, lstegglog}

\usepackage{hyperref}
\usepackage[noabbrev,nameinlink,capitalize]{cleveref} 

\usepackage{enumitem}

\usepackage{xspace}

\usepackage{nowidow}
\usepackage{microtype}

\usepackage{wrapfig}
\usepackage{float}
\usepackage{caption}
\usepackage{subcaption}

\usepackage{threeparttable}

\newcommand{\revemptyset}{\ensuremath{\text{\reflectbox{$\emptyset$}}}}
\newcommand{\emdash}{\---}
\usepackage{xcolor}
\definecolor{ltblue}{rgb}{0,0.4,0.4}
\definecolor{dkblue}{rgb}{0,0.1,0.6}
\definecolor{dkgreen}{rgb}{0,0.35,0}
\definecolor{dkviolet}{rgb}{0.3,0,0.5}
\definecolor{dkred}{rgb}{0.5,0,0}
\definecolor{dkorange}{HTML}{e87954}
\definecolor{dkyellow}{HTML}{f7cd7a}
\definecolor{dkpink}{HTML}{d152cf}

\newcommand{\vyzx}{\textsl{Vy}\textsc{ZX}\xspace}

\newcommand{\qlib}{\texttt{QuantumLib}\xspace}

\newcommand{\zxcalc}{ZX-calculus\xspace}
\newcommand{\zxdiag}{ZX-diagram\xspace}
\newcommand{\zxdiags}{\zxdiag{}s\xspace}

\newcommand{\tensrocq}{TensorRocq\xspace}

\newcommand{\C}{\mathbb{C}}

\usepackage{etoolbox} 
\newtoggle{comments}
\toggletrue{comments}

\iftoggle{comments}{
  \newcommand{\fixme}[1]{\textbf{\textcolor{red}{[ Fixme: #1]}}}
  \newcommand{\todo}[1]{\textbf{\textcolor{green}{[ TODO: #1 ]}}}
  \newcommand{\rnr}[1]{\textbf{\textcolor{blue}{[ Robert: #1 ]}}}
  \newcommand{\ben}[1]{\textbf{\textcolor{orange}{[ Ben: #1 ]}}}
  \newcommand{\wjbs}[1]{\textbf{\textcolor{purple}{[ William: #1 ]}}}
  \newcommand{\wjbscomment}[1]{} 
}{
  \newcommand{\fixme}[1]{}
  \newcommand{\todo}[1]{}
  \newcommand{\rnr}[1]{}
  \newcommand{\ben}[1]{}
  \newcommand{\wjbs}[1]{}
  \newcommand{\wjbscomment}[1]{}
}

\usepackage{tikz}

\usetikzlibrary{backgrounds}
\usetikzlibrary{arrows}
\usetikzlibrary{shapes,shapes.geometric,shapes.misc}
\usetikzlibrary{decorations.pathmorphing}
\usetikzlibrary{decorations.pathreplacing}
\usetikzlibrary{decorations.markings}

\tikzstyle{every picture}=[baseline=-0.25em,scale=0.5]

\pgfkeys{/tikz/tikzit fill/.initial=0}
\pgfkeys{/tikz/tikzit draw/.initial=0}
\pgfkeys{/tikz/tikzit shape/.initial=0}
\pgfkeys{/tikz/tikzit category/.initial=0}

\newcommand{\tikzfig}[1]{%
\IfFileExists{#1.tikz}
  {\input{#1.tikz}}
  {%
    \IfFileExists{./figures/#1.tikz}
      {\input{./figures/#1.tikz}}
      {\tikz[baseline=-0.5em]{\node[draw=red,font=\color{red},fill=red!10!white] {\textit{#1}};}}%
  }%
}
\newcommand{\ctikzfig}[1]{%
\begin{center}\rm
  \input{#1.tikz}
\end{center}}

\pgfdeclarelayer{edgelayer}
\pgfdeclarelayer{nodelayer}
\pgfsetlayers{background,edgelayer,nodelayer,main}
\tikzstyle{none}=[inner sep=0mm]
\tikzstyle{every loop}=[]
\tikzstyle{mark coordinate}=[inner sep=0pt,outer sep=0pt,minimum size=3pt,fill=black,circle]
\input{zh.tikzdefs}

\tikzstyle{dot}=[inner sep=0.3mm, minimum width=2mm, minimum height=2mm, draw, shape=circle, font={\footnotesize}, tikzit fill=magenta, fill=white]
\tikzstyle{white dot}=[dot, fill={rgb,255: red,232; green,165; blue,165}, text depth=-0.2mm, tikzit category=ZH-pf, draw=black]
\tikzstyle{white phase dot}=[minimum size=5mm, font={\footnotesize\boldmath}, shape=rectangle, rounded corners=2mm, inner sep=0.2mm, outer sep=-2mm, scale=0.8, tikzit shape=circle, draw=black, fill={rgb,255: red,232; green,165; blue,165}, tikzit category=ZH-pf, tikzit draw=blue]
\tikzstyle{gray dot}=[dot, fill={rgb,255: red,216; green,248; blue,216}, text depth=-0.2mm, tikzit category=ZH-pf]
\tikzstyle{gray phase dot}=[white phase dot, tikzit shape=circle, tikzit draw=blue, fill={rgb,255: red,216; green,248; blue,216}, font={\footnotesize\boldmath}]
\tikzstyle{hadamard}=[fill=white, draw, inner sep=0.6mm, minimum height=1.5mm, minimum width=1.5mm, shape=rectangle, tikzit shape=rectangle, tikzit category=ZH-pf]
\tikzstyle{small hadamard}=[hadamard]
\tikzstyle{lambda}=[hadamard, fill={rgb,255: red,180; green,180; blue,180}, tikzit shape=rectangle]
\tikzstyle{halfscalar}=[star, fill=black, draw=black, minimum size=8pt, inner sep=0pt]
\tikzstyle{box}=[shape=rectangle, text height=1.5ex, text depth=0.25ex, fill=white, draw=black, minimum height=3mm, minimum width=5mm, font={\small}]
\tikzstyle{Z dot}=[inner sep=0mm, minimum size=2mm, shape=circle, draw=black, fill={zx_green}, tikzit fill=green]
\tikzstyle{Z phase dot}=[minimum size=5mm, font={\footnotesize\boldmath}, shape=rectangle, rounded corners=2mm, inner sep=0.2mm, outer sep=-2mm, scale=0.8, tikzit shape=circle, draw=black, fill={zx_green}, tikzit draw=blue, tikzit fill=green]
\tikzstyle{X dot}=[Z dot, shape=circle, draw=black, fill={zx_red}, tikzit fill=red]
\tikzstyle{X phase dot}=[Z phase dot, tikzit shape=circle, tikzit draw=blue, fill={zx_red}, font={\footnotesize\color{black}\boldmath}, tikzit fill=red]
\tikzstyle{H box}=[hadamard]
\tikzstyle{st}=[star, star points=5, fill=white, draw=black, inner sep=1.2pt, line width=1.2pt, tikzit fill=blue, tikzit draw=red, tikzit category=ZH-pf]
\tikzstyle{triangle}=[regular polygon, regular polygon sides=3, fill=white, draw=black, inner sep=0pt, minimum width=1em, tikzit draw=blue, tikzit category=ZH-pf, tikzit fill=cyan]
\tikzstyle{not}=[fill={rgb,255: red,180; green,180; blue,180}, draw=black, shape=circle, font={$\neg$}, dot]
\tikzstyle{vertex}=[inner sep=0mm, minimum size=1mm, shape=circle, draw=black, fill=black]
\tikzstyle{vertex set}=[inner sep=0mm, minimum size=1mm, shape=circle, draw=black, fill=white, font={\footnotesize\boldmath}]
\tikzstyle{wide point}=[fill=white, draw, shape=isosceles triangle, shape border rotate=-90, isosceles triangle stretches=true, inner sep=0pt, minimum width=1.5cm, minimum height=6.12mm, yshift=-0.0mm]
\tikzstyle{medium gray box}=[semilarge box, fill={rgb,255: red,180; green,180; blue,180}]
\tikzstyle{small box}=[rectangle, inline text, fill=white, draw, minimum height=5mm, yshift=-0.5mm, minimum width=5mm, font={\small}]
\tikzstyle{small gray box}=[small box, fill={rgb,255: red,180; green,180; blue,180}]
\tikzstyle{medium box}=[rectangle, inline text, fill=white, draw, minimum height=5mm, yshift=-0.5mm, minimum width=8mm, font={\small}]
\tikzstyle{ddot}=[line width=1.6pt, inner sep=0mm, minimum width=2.5mm, minimum height=2.5mm, draw, shape=circle]
\tikzstyle{dd white}=[ddot, fill=white, tikzit draw=green]
\tikzstyle{dd white phase}=[white phase dot, line width=1.6pt, tikzit draw=yellow]
\tikzstyle{dd gray}=[ddot, fill={rgb,255: red,180; green,180; blue,180}, tikzit draw=green]
\tikzstyle{dd gray phase}=[gray phase dot, line width=1.6pt, tikzit draw=yellow]
\tikzstyle{cnotbot}=[fill=white, draw=black, shape=circle]
\tikzstyle{new edge style 0}=[->-]
\tikzstyle{Inner arrow}=[->-]
\tikzstyle{State Prep}=[fill=white, draw=black, shape=rounded rectangle, rounded rectangle east arc=0pt, tikzit shape=rectangle, tikzit category=Measurement, tikzit fill={rgb,255: red,247; green,0; blue,255}]
\tikzstyle{Measure}=[fill=white, draw=black, shape=rounded rectangle, rounded rectangle west arc=0pt, tikzit category=Measurement, tikzit fill=cyan, tikzit shape=rectangle]
\tikzstyle{Gate}=[fill=white, draw=black, shape=rectangle, tikzit category=Measurement, tikzit shape=rectangle]
\tikzstyle{Vertex Highlight}=[fill=orange, draw=orange, shape=circle]
\tikzstyle{medium box}=[fill=white, draw=black, shape=rectangle, minimum width=0.75cm, minimum height=1cm]
\tikzstyle{rotated text}=[rotate=270]

\tikzstyle{simple}=[-, fill=none]
\tikzstyle{hadamard edge}=[-, dashed, dash pattern=on 2pt off 1pt, thick, draw=blue]
\tikzstyle{gray}=[-, draw={blue!60!white}, tikzit draw=blue]
\tikzstyle{blue}=[-, draw={blue!60!white}, tikzit draw=blue]
\tikzstyle{brace edge}=[-, tikzit draw=blue, decorate, decoration={brace,amplitude=1mm,raise=-1mm}]
\tikzstyle{diredge}=[->]
\tikzstyle{not edge}=[-, dashed, dash pattern=on 2pt off 1.5pt, thick, draw={rgb,255: red,255; green,68; blue,68}]
\tikzstyle{double edge}=[-, double, shorten <=-1mm, shorten >=-1mm, double distance=2pt]
\tikzstyle{boldedge}=[-, line width=1.6pt, shorten <=-0.17mm, shorten >=-0.17mm, tikzit draw=blue]
\tikzstyle{separatediagrams}=[-, dashed, dash pattern=on 2pt off 1.5pt, thick, draw=black]
\tikzstyle{x stabilizer}=[-, double=black, draw={zx_red}, line width=1pt]
\tikzstyle{z stabilizer}=[-, double=black, draw={zx_green}, line width=1 pt]
\tikzstyle{zx stabilizer}=[-, double=black, draw={zx_yellow}, line width=1 pt]
\tikzstyle{measurement}=[-, tikzit fill=cyan, double]
\tikzstyle{dashed diredge}=[->, dashed, dash pattern=on 2pt off 1.5pt]

\newcommand{\pngfig}[2]{%
    \IfFileExists{#1.pdf}{\includegraphics[width=#2]{#1.pdf}}
    {
        \IfFileExists{./figures/#1.pdf}{\includegraphics[width=#2]{./figures/#1.pdf}}
        {
            \IfFileExists{#1.png}
              {\PackageWarning{pngfig}{#1.png is a PNG; slows down compile time; consider converting to pdf}
                \includegraphics[width=#2]{#1.png}}
              {
                \IfFileExists{./figures/#1.png}
                  {\PackageWarning{pngfig}{./figures/#1.png is a PNG; slows down compile time; consider converting to pdf}
                  \includegraphics[width=#2]{./figures/#1.png}}
                  {\tikz[baseline=-0.5em]{\node[draw=red,font=\color{red},fill=red!10!white] {\textit{#1}};}}%
              }
        }
    }
}

\newcommand{\pnghfig}[2]{%
    \IfFileExists{#1.pdf}{\includegraphics[height=#2]{#1.pdf}}
    {
        \IfFileExists{./figures/#1.pdf}{\includegraphics[height=#2]{./figures/#1.pdf}}
        {
            \IfFileExists{#1.png}
              {\PackageWarning{pnghfig}{#1.png is a PNG; slows down compile time; consider converting to pdf}
                \includegraphics[height=#2]{#1.png}}
              {
                \IfFileExists{./figures/#1.png}
                  {\PackageWarning{pnghfig}{./figures/#1.png is a PNG; slows down compile time; consider converting to pdf}
                  \includegraphics[height=#2]{./figures/#1.png}}
                  {\tikz[baseline=-0.5em]{\node[draw=red,font=\color{red},fill=red!10!white] {\textit{#1}};}}%
              }
        }
    }
}

\newtheorem{lem}{Lemma}

\newtheorem{dfn}{Definition}

\newcommand{\lstlineref}[2]{\hyperref[#2]{\Cref*{#1}:\ref*{#2}}}
\newcommand{\lstlinerangeref}[3]{\hyperref[#2]{\Cref*{#1}:\ref*{#2}-\ref*{#3}}}
\crefname{lem}{Lemma}{Lemmas}
\Crefname{lem}{Lemma}{Lemmas}
\crefname{cor}{Corrolary}{Corrolaries}
\Crefname{cor}{Corrolary}{Corrolaries}
\crefname{thm}{Theorem}{Theorems}
\Crefname{thm}{Theorem}{Theorems}
\crefname{def}{Definition}{Definitions}
\Crefname{def}{Definition}{Definitions}
\crefname{ex}{Exercise}{Exercises}
\Crefname{ex}{Exercise}{Exercises}
\crefname{exa}{Example}{Examples}
\Crefname{exa}{Example}{Examples}

\usepackage{pgfplots} 
\pgfplotsset{compat=1.18}

\usepackage[T1]{fontenc}

\title{\tensrocq: Enabling diagrammatic reasoning in Rocq}
\author{Benjamin Caldwell} 
\affiliation{%
\institution{University of Chicago}
\country{USA}
}
\author{William Spencer}
\affiliation{%
\institution{University of Chicago}
\country{USA}
}
\author{Aleks Kissinger}
\affiliation{%
\institution{University of Oxford}
\country{UK}
}
\author{Robert Rand}
\affiliation{%
\institution{University of Chicago}
\country{USA}
}
\date{\today}

\begin{document}

\begin{abstract}

Symmetric monoidal categories (SMCs) are a common framework for reasoning about computation, focusing on the parallel and sequential compositionality of operations. String diagrams are a ubiquitous and powerful tool for reasoning about equations in SMCs, leveraging eliding the fine details of compositionality to focus on connectivity. However, when working with SMCs in a proof assistant, the rigid equational structure of composition occludes the essential connective information, longer proofs filled with uninformative syntactic manipulation. To address the gap between proof assistants and paper proof, we have developed verified tools for diagrammatic reasoning in Rocq, including inferring term equivalence and rewriting modulo the deformation of string diagrams. This is achieved by converting between syntactic representations of SMC terms and hypergraphs with interfaces, while preserving a common tensor semantics. We provide tools to develop simple SMC theories from generators and relations, and perform equational reasoning these systems. We also enable our tactics to be used in existing verification projects about SMCs which can be given semantics as tensor expressions.

\end{abstract}

\maketitle

\section{Introduction}

Process theories are representations of computation as a operation which is \textit{composable}, in sequence and in parallel~\cite{coecke2017picturing}. 
Many systems can be modeled as process theories, including logic circuits, the ZX-calculus~\cite{CoeckeDuncan2011}, and even linear algebra, by viewing matrices as transforming a vector space. 
From a category-theoretic perspective, process theories are instances of Symmetric Monoidal Categories (SMCs).
Through this interpretation, process theories gain a natural diagrammatic language of ``string diagrams.'' 
Much like a flowchart, string diagrams represent the structural relationship between components of a process, usually interpreted as the flow of information. 
These diagrams can then be enhanced with an equational theory specifying processes which should be considered equivalent.
Often, these equations arise from equalities of the semantics of a process theory, such as the matrix formalism of quantum circuits.

Process theories are easy to reason about on paper because reasoning corresponds to diagrammatic manipulation.
String diagrams place a particular focus on the flow of data through a process, which is encoded by \textit{connectivity}.
In fact, two string diagrams with the same connections represent the same process, as encapsulated by the mantra of ``only connectivity matters''~\cite{Selinger2010}.
Within the categorical framework of SMCs, this is encoded via a number of naturality and coherence conditions that express that the choice of associativity within a term does not affect its interpretation.
String diagrams encode all these facts in the single notion that a diagram may be freely deformed without changing its meaning.
The pattern of reasoning about a string diagram on paper is therefore a cycle of identifying some subdiagram which you want to change, deforming the diagram to separate that subdiagram, and replacing it with an equivalent diagram.

When attempting to formalize process theories within a proof assistant, things become more complex.
Associativity information that does not affect connectivity can be elided on paper because we know it does not affect the represented process.
However, representing an SMC in a proof assistant by an inductive datastructure requires a choice of associativity to be made explicitly. To show two terms are equivalent, it is necessary to apply naturality and coherence conditions explicitly to make the terms syntactically identical.
Rewriting with known equalities is just as awkward, as the target diagram must be manually deformed using only the naturality and coherence conditions to display the desired subdiagram syntactically, often requiring a large number of structural rewrites.
In practice, this means that reasoning about process theories in a proof assistant is the cycle of mentally decoding the visual noise of associativity to understand the connective information of a diagram, eventually identifying a subdiagram you want to change, tediously altering the associativity of the diagram until the subdiagram appears verbatim, and finally replacing the diagram.
The vast majority of proof lines can easily become dedicated to appeasing the proof assistant of the same conditions over and over.

Prior work has developed some techniques for how to ease the burden of dealing with associativity in SMCs.
Some tools implement automation to move terms around~\cite{vicar}, while other solutions use external tooling to automatically solve these associations~\cite{pous2026string}.
Of particular interest to our project is the tool Chyp~\cite{chyp}, which solves this problem by translating SMC terms into a graph structure to perform rewrites and extracting the resulting graph back to an SMC.
However, Chyp cannot be used within a general-purpose proof assistant because (1) it performs purely equational reasoning and cannot handle first- or higher-order logic and (2) its tactics and proof checker are based on graph rewriting, and thus cannot produce proofs in a format that could be accepted by a traditional interactive theorem prover.

There are major challenges in embedding the rewriting system of a tool such as Chyp fully within a proof assistant, as the existing tooling is unverified. 
Hence, we must extend the methodology behind Chyp to include denotational semantics to justify its reasoning chain.
This is a complex undertaking, as several interacting theories and the translations between them must be designed in a way that is amenable to verification. 
Moreover, performance is key when designing a rewrite tactic as it may be used many times throughout a single proof. 
To maximize performance, our development extensively uses reflective automation, which proves goals through the use of verified algorithms.
Due to this decision, we had to carefully choose the data structures we used to ensure they were performant enough to yield fast automation; in this we relied heavily on the \texttt{stdpp} library for fast implementations. 
We build the theory for these denotational semantics in \cref{sec:theory} and cover the design and proof engineering decisions that went into giving their verified implementations in \cref{sec:implementation}.


To this end, we present \href{https://github.com/inqwire/tensorrocq}{\tensrocq}, a Rocq tool for verified diagrammatic reasoning about SMCs with respect to tensor semantics.  
This implementation showcases a number of novel features:
\begin{itemize}
    \item We automatically resolve equalities derivable from the axioms of SMCs, i.e. only connectivity matters, allowing us to ignore associativity within SMC terms.
    \item We provide a tactic for diagrammatic rewriting within SMCs, automatically finding and replacing subdiagrams matching proven equalities up to associativity.
    \item We create a framework for easily defining new SMC theories in the style of Chyp, which are given semantic meaning 
    \item We provide a unifying framework to verify the relationship between SMCs and hypergraphs via tensor semantics.
    \item We provide an extensible system of typeclasses allowing \tensrocq to be integrated into existing projects by instantiating SMC theories.
    \item We provide the flexibility to work with existing Rocq libraries without requiring libraries to change their existing definitions or statements of lemmas.
\end{itemize}

By using the rewrite engine of \tensrocq, paper-style diagrammatic proofs can be directly translated into formal proofs with associativity elided automatically.
Focusing on the meaningful rewrites over the noise of associativity makes proofs more readable, because they are no longer dominated by syntactic manipulation.
It also makes proofs more resilient to small changes in definitions or statement, as terms are considered up to full SMC equivalence.

\section{Related Work}

Verification of symmetric monoidal categories, and associated rewriting tools, has been a popular topic of research in recent years.
We focus on three projects that tackle this problem directly: ViCAR~\cite{vicar}, Chyp~\cite{chyp}, and String Diagrams for Monoidal Categories in Rocq~\cite{pous2026string}.

The ViCAR project aims for generality by implementing typeclasses and using monoidal coherence to assist in rewriting diagrams~\cite{vicar}. It is embedded within Rocq and uses a technique called \emph{foliation} to separate terms into a normal form where patterns can be more easily identified. It provides a basic visualization engine, but is unable to ignore associativity constraints in composing morphisms, so the visualization has to directly convey associativity information, unlike in string diagrams.

Chyp is a proof assistant for symmetric monoidal categories written in Python~\cite{chyp}. It implements a cospan of hypergraphs data structure to rewrite symmetric monoidal category terms. Chyp has a focus on interactivity, as a standalone proof assistant that can be used to prove simple facts about SMCs defined by generators and relations~\cite{chyp}. A Chyp theory consists of atomic generators, definitions, axioms, and  theorems, which are proven using transitive chains of diagram equations solved by a handful of built-in tactics. Theories can be combined via a rudimentary module system with ML-style generic parameters, allowing one to build more complex theories from simpler ones.

Pous' String Diagrams project has a fully interactive string diagram editor which can extract the necessary steps for a Rocq proof~\cite{pous2026string}. This project is situated at the intersection of the above works. Like ViCAR, it attempts to be generally applicable to Rocq projects with its companion Rocq library. In order to do interactive proof, String Diagrams for Rocq uses an external tool that can import Rocq terms and export Rocq proofs as a series of tactic calls. 

Of these tools, \tensrocq is most similar to Chyp. It shares Chyp's theoretical underpinning in hypergraphs (described in a series of papers by Bonchi et al.~\cite{bonchi1, bonchi2, bonchi3}), adapted and extended to work with the Rocq proof assistant. This adaptation involved adding a semantic underpinning through tensors in order to verify the rewriting process undertaken by Chyp, allowing the rewrite engine to be used in any existing project that can be given semantics via tensors.

\section{Theoretical Foundations}\label{sec:theory}

\tensrocq is a Rocq library for diagrammatic reasoning in SMCs based on the formal connection between hypergraphs, PROPs, and tensors.
Specifically, hypergraphs, PROPs, and SMCs can be equipped with semantic interpretations as tensor expressions, which grounds translations between them in a common basis that will serve as the foundation for their application to formal verification.
In this section, we review these three building blocks and elaborate on the connections between them. In the next section, we give their Rocq definitions and details of their implementation.


\subsection{Tensors}

\begin{figure}[ht]
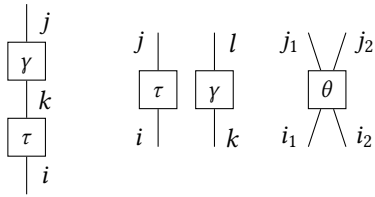

    \centering
    \ctikzfig{tensorvis}
    \caption{A visualization of tensor terms, including contraction, product, and tensors with multiple inputs. Following the standard convention, diagrams should be read from bottom to top.}
    \label{fig:tensorvisualization}
\end{figure}

Tensors are a popular abstraction in classical and quantum physics, where a tensor is given as a set of indexed numbers over a field. The following definitions related to tensors draw on Kissinger's work~\cite{Kissinger2014}, which will we discuss in the context of proof assistants.

\begin{dfn}[Concrete Tensor]
    A tensor $\tau$ over a field $\mathbb{F}$ with $n$ inputs and $m$ outputs with domain $T$ gives a value
    \[ \tau_{i_1\dots i_n}^{o_1\dots o_m} \in \mathbb{F} \]
    for all choices of $i_1,\dots,i_n, o_1,\dots,o_m \in I$, for some finite indexing set $I$. This can be equivalently represented as $\tau_{i}^o$ where $i$ and $o$ are vectors of lengths $n$ and $m$ respectively. 
\end{dfn}

The simplest example of a tensor would be the delta tensor. The delta tensor, often called the Kronecker delta, is 1 when its inputs and outputs are precisely equal and 0 otherwise. As a process, it can be thought of as an identity operator, as it does not change its outputs. As matrices, the identity matrix is the Kronecker delta.

\begin{dfn}[Delta tensor]\label{dfn:deltatensor}
    For an arity $n$ and field $\mathbb{F}$, the delta tensor $\delta$ over $F$ with $n$ inputs and $m$ outputs is defined so $\delta_{i}^{o}$ is $1$ if $i=o$ and $0$ otherwise. 
\end{dfn}

We then also define a process called tensor contraction. This allows us to give some notion of process flow for tensors. When connecting an output to an input, contraction acts as a sequencing operator, moving outputs from one tensor to inputs of another. When connecting inputs between two tensors, it restricts those inputs to always be equal. The same is true for outputs. The notion of requiring inputs or outputs to always be equal is a strange idea for processes, but our goal is to use tensors as a semantic domain, so this extra expressivity is not a problem.

\begin{dfn}[Tensor Contraction]\label{dfn:contraction}
    The contraction of two tensors $\sum_{k_1,\in I}\tau_{i}^{k}\gamma_{k}^{j}$ is taken by summing over the indexing set $I$. 
\end{dfn}

If two tensors are interpreted as modeling the behavior of some processes, contraction models the behavior of the composition of these processes.
For example, the delta tensor models the identity process, whose output is precisely its input, and this is indeed an identity for contraction.
Tensors also offer a way to reason about parallel composition by taking the product of two tensors in the corresponding field. 

\begin{dfn}[Tensor Product]
    The product of two tensors $\tau_i^j$ and $\gamma_k^\ell$ is given by $(\tau\gamma)_{ik}^{jl} = \tau_i^j * \gamma_k^l$, where $*$ is multiplication in the underlying field.
\end{dfn}

These tensors should be thought of as black-box operations. Many examples exist, notably vector spaces over a field $\mathbb{F}$ can be represented as tensors. The tensor product is then the Kronecker product on matrices, while composition is matrix multiplication. Tensors are a useful tool for reasoning diagrammatically. As tensors have a natural interpretation as traced-symmetric monoidal categories, they have a natural diagrammatic interpretation~\cite{Kissinger2014}.

There is also a dual syntactic notion which avoids writing out tensor contraction explicitly, and says a tensor contraction occurs when two variables are shared. \cref{dfn:contraction} could then be written as $\tau_i^k\gamma_k^j = \sum_{k\in I} \tau_i^k\gamma_k^j$. We also can visualize these terms as in \Cref{fig:tensorvisualization}.

Tensors are particularly useful for verification. They are a heavily abstracted form of linear algebra which has a natural way to express both composition and connectivity through the same operator. This means they can connect to the other parts of our rewrite engine, hypergraphs and APROPs, which we will see in \cref{subsec:bg-hypergraphs} and \cref{subsec:bg-props}. Tensors have also been shown to have a natural interpretation as monoidal categories, further cementing their significance in the intersection between hypergraphs and APROPs~\cite{Kissinger2014}.

\subsection{Hypergraphs}\label{subsec:bg-hypergraphs}


Our definitions for hypergraphs build on those of Bonchi~\cite{bonchi2}, extending them with an interpretation of hypergraphs as tensors.
A hypergraph is an extension of a graph where an edge can be incident to multiple vertices. 
In graphs, vertices are usually thought of as ``objects'' that edges can connect.
In hypergraphs, it is common to flip that hierarchy, instead thinking of \textit{edges} as primary, with vertices describing how edges are connected.
As such, we often wish to equip edges with some data, which can be used to interpret a hypergraph in a concrete way.

\begin{dfn}[Labeled Hypergraph]
    A labeled hypergraph is an ordered pair containing a set of vertices $\mathcal{V} = \{v_0 \dots v_i\}$ and edges $\mathcal{E} = {(l_0, e_0) \dots (l_j, e_j)}$ where $e_k \subseteq V$ for all k and $l_k$ is the label of the hyperedge. 
\end{dfn}

With this definition, a graph is an instance of a hypergraph where each $e_k$ has size 2. 
Labeled hypergraphs have a natural interpretation as tensors if the labels are associated to tensors, with vertices describing what indices to sum over in the resulting tensor. 
To make this translation concrete, it is easiest to work with \textit{directed} hypergraphs, whose edges' vertex sets are split into inputs and outputs, as tensors have designated inputs and outputs. 

\begin{dfn}[Labeled Directed Hypergraph]
    A labeled directed hypergraph is an ordered pair containing a set of vertices $\mathcal{V} = \{v_0 \dots v_i\}$ and edges $\mathcal{E} = {(l_0, el_0, er_0) \dots (l_j, el_j, er_j)}$ where $el_k, er_k \subseteq V$ for all k and $l_k$ is the label of the hyperedge. 
    $e_l$ is called the left edge set, or inputs, of the edge, while $e_r$ is the right edge set, or outputs, of the edge.
\end{dfn}

\subsubsection{Hypergraphs with Interfaces}

Hypergraphs are not sufficient to model process theories, as we often think of a process as having inputs and outputs. 
To extend our notion of hypergraphs, we add the notion of an interface. 
Interfaces act as the inputs and outputs of the process represented by the hypergraph as a whole.
Chyp refers to these interfaces by their categorical name of \emph{cospans} and Bonchi et al. give their categorical interpretation~\cite{bonchi2}.
We use the term interfaces to emphasize their interpretation over their categorical meaning.

\begin{dfn}[Interface] 
    An interface for a hypergraph $H$ is a pair of left and right ordered lists of vertices $(V_L, V_R)$ where $V_L$ are considered the ``inputs'' to the graph and $V_R$ are considered the outputs. If $v \in V_L$ we have a single input into the vertex $v$. The order of the inputs and outputs matters, and we will take a top-down approach to ordering edges from our interfaces.
\end{dfn}

\begin{figure}
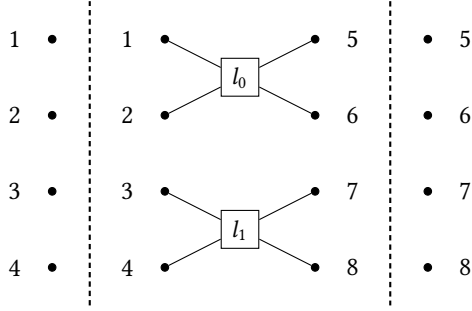

    \centering
    \ctikzfig{interfaces}
    \caption{A labeled directed hypergraph with interfaces on the left and right. The interfaces fix the order on the inputs to hyperedges, which may differ from the order in their adjacent hyperedges.}
    \label{fig:interfaces}
\end{figure}

We will often write our hypergraph with interfaces as an ordered triple given by $(V_L, H, V_R)$ where the inputs are on the left, outputs on the right, and the graph in the middle. With these new interfaces, we need an updated definition for degree which respects the idea that vertices can have connections to the interfaces as well.

\begin{dfn}[Tensors of Labeled Hypergraphs with Interfaces]
    Let $\mathcal{H}=(\mathcal{V}, \mathcal{E})$ be a labeled directed hypergraph with interfaces $V_L$ and $V_R$, and for each edge $e_k=(l_k, el_k, er_k)$ let a tensor $\tau(e_k)$ be given with $|el_k|$ inputs and $|er_k|$ outputs, over a field $\mathbb{F}$ and finite indexing set $A$.
    Then, the tensor semantics of $\mathcal{H}$ is
    \[ \tau(\mathcal{H})_{i}^{o} =
        \sum_{f : V \rightarrow A} 
        \delta_{i}^{f(V_L)} \cdot
        \delta_{f(V_R)}^{o} \cdot
        \prod_{(l_k,v,w)\in \mathcal{E}}{\tau(e_k)_{f(v)}^{f(w)}}.
        \]
\end{dfn}

The key feature of tensors and hypergraphs is their emphasis on \textit{connectivity} of tensors.  The benefit of using the two in combination is that they have complementary characteristics.  Tensors are abstract and not syntactic, but have very clear semantic meaning as functions to fields, while hypergraphs are concrete and syntactic objects but lack the same immediate semantic interpretation. Combining the two allows syntactic reasoning about hypergraphs to prove algebraic facts, as is detailed in~\cite{Kissinger2014}.

The reverse direction is true as well: every tensor expression (meaning a tensor built up with product and contraction from some set of known, abstract "black-box" tensors) can be represented by a labeled directed hypergraph by interpreting each input, output, and contraction as a vertex, and each abstract tensor as an edge from its inputs to its outputs. The power of this translation is that labeled directed hypergraphs which are isomorphic have equal semantics as tensor expressions. This allows algebraic relations to be proven by graphical reasoning. To show this, we build sequential composition and parallel products of hypergraphs to mirror the contraction and product of their tensor semantics.

For simplicity, composition is defined only for hypergraphs with interfaces which have equal interior interfaces, and with only these vertices in common. This can always be ensured by relabeling one of the hypergraphs.

\begin{dfn}[Vertices]\label{dfn:vertices-interfaces}
    The set of vertices $V(\mathcal{H})$ of a hypergraph with interfaces $\mathcal{H}=(V_L,H,V_R)$, where $H=(V,E)$, is the set $V_L\cup V_R \cup V$.
\end{dfn}

\begin{dfn}[Composition of Interfaces]
    Given two hypergraphs with interfaces $\mathcal{H}=(H_L,H,H_R)$ and $\mathcal{J}=(J_L,J,J_R)$, where $H_R = J_L = V(\mathcal{H}) \cap V(\mathcal{J})$, we define the composition $\mathcal{H};\mathcal{J}=(H_L, H, H_R) ; (J_L, J, J_R)$ of these hypergraphs to be $(H_L, H\,\sqcup\, J , J_R)$. Here, $H\,\sqcup\, J$ stands for the hypergraph whose vertices are the union of the vertices of $H$ and $J$ and whose edges are the disjoint union of the edges of $H$ and $J$. 
\end{dfn}

This definition of composition joins hypergraphs along a common interface, and corresponds directly to the contraction of the corresponding tensor semantics.

\begin{lem}
\label{lem:compose}
    Let $\mathcal{H}=(H_L,H,H_R)$ and $\mathcal{J}=(J_L,J,J_R)$ be hypergraphs with interfaces such that $H_R = J_L = V(\mathcal{H}) \cap V(\mathcal{J})$. 
    For each edge $e=(l,el,er)$ belonging to either $H$ or $J$, let a tensor $\tau(e)$ be given with $|el|$ inputs and $|er|$ outputs, over a field $\mathbb{F}$ and finite indexing set $A$. Let $n=|H_L|$, $m=|H_R|=|J_L|$, and $o=|J_R|$.
    Then
    \[ \tau(\mathcal{H};\mathcal{J})_{i}^{k} = \sum_{j} {\tau(\mathcal{H})_{i}^{j}\tau(\mathcal{J})_{j}^{k}}. \]
    
\end{lem}
\begin{proof}
        Let $H=(V,E)$ and $J=(U,F)$, so $V\cap U = J_L = H_R$.
    By definition, the tensor semantics of the composition is 
    \[ \tau(\mathcal{H};\mathcal{J})_{i}^{k} = 
        \sum_{f : V \cup U \rightarrow A} 
        \delta_{i}^{f({H_L})} \cdot
        \delta_{f({J_R})}^{k} \cdot
        \prod_{(l_k,v,w)\in E \sqcup F}{\tau(e_k)_{f(v)}^{f(w)}}
    \]
    
    The contraction of the respective tensor semantics is
    \begin{equation}
    \begin{split}
        & {\tau(\mathcal{H})_{i}^{j}\tau(\mathcal{J})_{j}^{k}} =
        \\ &
        \sum_{f_V : V \rightarrow A}
        \delta_{i}^{f_V({J_L})} \cdot
        \delta_{f_V({J_R})}^{j} \cdot
        \prod_{(l_k,v,w)\in E}{\tau(e_k)_{f(v)}^{f_V(w)}} \cdot
        \sum_{f_U : U \rightarrow A}
        \delta_{j}^{f_U({H_L})} \cdot
        \delta_{f_U({H_R})}^{k} \cdot
        \prod_{(l_k,v,w)\in F}{\tau(e_k)_{f(v)}^{f_U(w)}}
    \end{split}
    \end{equation}
        
    Observe that the functions $f_V$ and $f_U$ are each required to take values $j_1,\dots,j_m$ on $H_R=J_L$ (rather, if they do not, the summand is zero because of the delta tensors). Moreover, recall that $V$ and $U$ are disjoint outside $H_R=J_L$; so, we can join these summations to run over $f:V\cup U \rightarrow A$ substituting for $f_V$ and $f_U$, removing the delta tensor terms involving $H_R$ and $J_L$. Then, 
    observing that
    \begin{equation}
    \begin{split}
    & \prod_{(l_k,(v_1,\dots,v_n),(w_1,\dots,w_m))\in E \sqcup F}{\tau(e_k)_{f(v_1),\dots,f(v_n)}^{f(w_1),\dots,f(w_m)}} =
    \\ &
        \prod_{(l_k,(v_1,\dots,v_n),(w_1,\dots,w_m))\in E}{\tau(e_k)_{f(v_1),\dots,f(v_n)}^{f(w_1),\dots,f(w_m)}} \cdot 
        \prod_{(l_k,(v_1,\dots,v_n),(w_1,\dots,w_m))\in F}{\tau(e_k)_{f(v_1),\dots,f(v_n)}^{f(w_1),\dots,f(w_m)}},
    \end{split}
    \end{equation}
    we arrive at (1).
\end{proof}

\begin{dfn}[Product of Interfaces]
   Given two disjoint hypergraphs $H$, and $J$, with disjoint interfaces $H_L, H_R$ and $J_L, J_R$ we define the parallel product, or stack, $(H_L, H, H_R) * (J_L, J, J_R)$ to be $(H_L + J_L, H\,\cup\, J, H_R + J_R)$. 
\end{dfn}

An analogous statement to that of \Cref{lem:compose} holds relating the tensor semantics of the product of hypergraphs with interfaces to the product of the tensor semantics of the hypergraphs. 

\begin{lem}
\label{lem:stack}
    Let $\mathcal{H}$ and $\mathcal{J}$ be disjoint hypergraphs with disjoint interfaces.
    Then 
    \[ \tau(\mathcal{H}\ast\mathcal{J})_{ij}^{kl} = 
    \tau(\mathcal{H})_{i}^{k} \tau(\mathcal{J})_{j}^{l}
    \]
\end{lem}

\begin{proof}
    Let $H = (V,E)$ with interfaces $V_L, V_R$ and $J = (U,F)$ with interfaces $U_L, U_R$ be disjoint hypergraphs and interfaces. Then $\tau(\mathcal{H} * \mathcal{J})^{kl}_{ij}$ is precisely given by
    \[ \tau(\mathcal{H} * \mathcal{J})_{ij}^{kl} =
        \sum_{f : V \cup U \rightarrow A} 
        \delta_{ij}^{f(I_i)f(I_j)} \cdot
        \delta_{f(O_k)f(O_l)}^{kl} \cdot
        \prod_{(l_k,v,w)\in E \cup F}{\tau(e_k)_{f(v)}^{f(w)}}.
        \]
    We use the fact that $\delta_{ij}^{kl} = \delta_i^k\delta_k^l$ as well as the fact that $V, U$ are disjoint, we can split the function $f$ into its restriction over $V$ and $U$ and split the sum. Similarly we can also split the product as $E, U$ are disjoint, giving us the term
    \[  \sum_{f : V \rightarrow A} 
        \delta_{i}^{f(I_i)} \cdot
        \delta_{f(O_k)}^{k} \cdot
        \prod_{(l_k,v,w)\in E}{\tau(e_k)_{f(v)}^{f(w)}} \cdot
        \sum_{f : U \rightarrow A} 
        \delta_{j}^{f(I_j)} \cdot
        \delta_{f(O_l)}^{l} \cdot
        \prod_{(l_k,v,w)\in F}{\tau(e_k)_{f(v)}^{f(w)}}.
        \]
    which is precisely 
    \[
        \tau(\mathcal{H})\cdot\tau(\mathcal{J}).
    \]
\end{proof}

\subsubsection{Rewriting Hypergraphs}

Rewriting at the level of hypergraphs is well explored. The general context is that there is a known rewrite rule $T \equiv S$ at the level of hypergraphs, and we want to substitute $T$ for $S$ in a larger hypergraph $H$ containing $T$ as a subgraph. The standard method to accomplish this is to represent $H$ as a \textit{double pushout}, in which it is separated into three parts, the middle one being isomorphic to $S$. Then, we can directly substitute $S$ in place of $T$.

\begin{lem}[Decomposition]\label{lem:decomposition}
    For a hypergraph $H$ with interfaces $V_L$ and $V_R$ and a given subhypergraph $T$ of $H$, we can form the decomposition of $H$ into $C_L$ and $C_R$ such that there exist $V_i, V_j,$ and $V_k$ such that 
        $H = (V_L, C_L, V_k + V_i) ; ((V_k, \emptyset, V_k) * (V_i, T, V_j)) ; (V_k + V_j, C_R, V_R)$
\end{lem}

\begin{figure}[h]
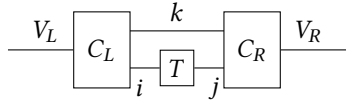

    \centering
    \ctikzfig{decomposition}
    \caption{The decomposition given by \Cref{lem:decomposition}.}
    \label{fig:decomposition}
\end{figure}

\begin{lem}[Double Pushout Rewrite]\label{lem:dporewrite}
    For a hypergraph $H$ with interfaces $V_L$ and $V_R$ and a rewrite rule $T = T'$ where $T$ is a subgraph of $H$, we can perform the decomposition procedure as in \cref{lem:decomposition}, replacing the term $T$ with $T'$.
\end{lem}

Both \Cref{lem:decomposition} and \Cref{lem:dporewrite} have been well explored and proved through the language of cospans of hypergraphs~\cite{bonchi2}. Their inclusion here serves to illustrate the process undertaken by \tensrocq.

\subsection{PROPs}\label{subsec:bg-props}

Diagrammatic reasoning and symmetric monoidal categories (SMCs) are closely tied together, with each SMC coming equipped with a diagrammatic language. In general, SMCs are systems with objects and transformations between them called arrows. We can take tensor products over objects and arrows in order to perform actions in parallel, and we can sequence arrows together to perform multiple actions in a row. For a clear definition of SMCs and their corresponding diagrammatic languages, we recommend Selinger's survey of the topic~\cite{Selinger2010}.

\begin{dfn}[PROPs]
    A prop is a symmetric monoidal category where the objects are natural numbers and the tensor product on objects is addition. Specifically we have a collection of arrows labeled as $f : n \to m$ where $n$ is the domain and $m$ the codomain of $f$, a braid operator $\beta^n_m : n \otimes m \to m \otimes n$, and composition $f ; g : n \to o$ for arrows $f : n \to m$ and $g : m \to o$.
\end{dfn}

PROPs are useful for tracking generic amounts of nonspecific information, making them perfect for process theories. The object $n$ for some natural number $n$ could be used to represent $n$ booleans and the arrows could be boolean operators such as $\mathit{AND}$, $\mathit{OR}$, and $\mathit{XOR}$. In general we can use them to reason about any kind of information flowing through a system, not just boolean circuit logic. Occasionally we want to put restrictions on two objects being equal in some sense. One way to do this is through the use of Autonomous PROPs or APROPs. 

\begin{dfn}[Autonomous PROP]\label{dfn:autonomous}
    Autonomous PROPs are PROPs enhanced with two families morphisms called the unit ($\eta_n$) and counit ($\epsilon_n$) given by $\eta_n : 0 \to n + n$ and $\epsilon_n : n + n \to 0$ such that $\eta_n \otimes id_n ; id_n \otimes \epsilon_n = id_n$ and $id_n \otimes \eta_n ; \epsilon_n \otimes id_n = id_n$.
\end{dfn}

What the $\eta$ and $\epsilon$ operators actually stand for becomes more clear in practice when we turn them into tensors or hypergraphs. As a hypergraph with interfaces, $\eta_n$ as a tensor is a collection of delta tensors (\cref{dfn:deltatensor}) which force the top $n$ outputs to be equal to the bottom $n$ outputs in order. The same is true of $\epsilon_n$, only for inputs. 
This allows for representing connectivity information which does not respect the directionality of edges.
Though these terms are not present in every SMC, they are used extensively in theories that include them.
By expressing them syntactically, we can reason about them directly at the level of hypergraph isomorphism, rather than treating them as generators and having to use extra rules.

APROPs are the primary syntactic interface for working with TensorRocq. All theories are stated as APROPs and then translated into hypergraphs to perform rewrites before being extracted back to APROP terms. With that in mind, we present the details of TensorRocq's implementation before proceeding to give examples of its use.

\section{Implementation}\label{sec:implementation}

To reason with hypergraphs in a verified context, we have to encode the concepts of hypergraphs, tensors, and PROPs within Rocq.
We want these definitions to be sufficiently \textit{expressive} to apply to any possible theory admitting tensor semantics.
We also want structures which are \textit{computational}, meaning we can perform operations or comparisons \textit{effectively} within Rocq\footnote{
By \textit{effective} computability, we mean that all functions and operations should be defined in such a way that Rocq's reduction engine can fully evaluate them, without being blocked by opaque definitions or axioms. 
This is a requirement for proof by reflection, which relies crucially on the reduction engine. 
}.
This unlocks the powerful technique of \textit{proof by reflection}, where verified algorithms are used to generate proofs.
Computation is usually much faster than proof search, and benefits from much better scaling to large goals, so is the go-to method for (semi-)decidable propositions. 
However, the most expressive definitions in Rocq are invariably not computable, as they embed the types and functions of Rocq directly, and these cannot in general be compared for equality.

The approach we take in \tensrocq is to use explicit, computational definitions as much as possible without limiting expressiveness.
For example, we restrict hypergraphs to have vertices indexed by binary natural numbers (Rocq's \coqe{positive} datatype), rather than any generic type, because this adds no restriction for finite graphs and positives allow for efficient computation (as opposed to Rocq's unary naturals).
We often recover expressiveness using automation, particularly through \textit{quotation}, in which we automatically find a syntactic, computable representation of some object that may not be otherwise effectively computable.
For example, this allows us to reason about PROPs with sets of generators that lack decidable equality by representing the finitely-many generators appearing in any given context by elements of some computable type. 

The \tensrocq rewrite engine is built to rewrite APROP terms by representing them with hypergraphs, and using effective algorithms to find matches up to SMC structure. 
We connect APROPs and hypergraphs with a common tensor semantics to give a definition of equivalence that enables our rewrites to be verified.
We extend the rewrite engine to apply within general SMCs, which may not have computational structure that algorithms can examine, by converting terms to APROPs.
In this way, APROPs serve as a syntactic interface to apply hypergraph-based reasoning within general SMCs, all verified with respect to tensor semantics.

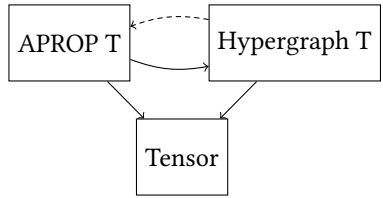
\begin{figure}[ht]
    \centering
    \begin{tikzpicture}
	\begin{pgfonlayer}{nodelayer}
		\node [style=medium box] (0) at (-3, 0) {APROP T};
		\node [style=medium box] (1) at (3, 0) {Hypergraph T};
		\node [style=medium box] (2) at (0, -3) {Tensor};
	\end{pgfonlayer}
	\begin{pgfonlayer}{edgelayer}
		\draw [style=diredge, bend right=15] (0) to (1);
		\draw [style=diredge] (0) to (2);
		\draw [style=diredge] (1) to (2);
		\draw [style=dashed diredge, bend left=345] (1) to (0);
	\end{pgfonlayer}
\end{tikzpicture}
    \caption{The core rewrite engine of \tensrocq, showing how APROPs and hypergraphs are equivalent and use tensors for semantics. Solid lines denote total functions, while the dotted line indicates a partial function supported on acyclic monogamous hypergraphs.}
    \label{fig:conversion}
\end{figure}

\subsection{Tensors}

The efficacy of tensors as a semantic backing is their \textit{compositionality} and their representation of \textit{connectivity}.
Tensors naturally represent sequential and parallel compositions by contraction and product, respectively, so they can encode the horizontal and vertical compositions of a symmetric monoidal category.
This is particularly useful because they encode the complex notion of ``only connectivity matters'' for string diagrams as simple algebraic identities, dramatically simplifying reasoning.
This simple equational theory makes reasoning about general tensors significantly simpler than reasoning about more complicated structures like SMC terms, particularly with the automatic solvers over (semi)rings provided by Rocq~\cite{ringtactic}.

\begin{lstlisting}[label=lst:tensordef, caption=Rocq definition of tensors]
Definition Tensor (R : Type) (n m : nat) (A : Type) := vec A n -> vec A m -> R.
Definition DimensionlessTensor (R : Type) (A : Type) := forall n m, Tensor R n m A.
\end{lstlisting}

To define tensors, we chose to use a shallow embedding using Rocq's functions directly, presented in \Cref{lst:tensordef}.
Because a tensor can be any function, we cannot perform any computation on its structure, but we obtain a sufficiently general definition to apply it to a variety of theories. 
To expose the compositional nature of tensors, we explicitly expose their sizes.
This ensures that contractions are well-typed, without having to provide any additional well-formedness conditions.
We found it necessary to work with ``dimensionless'' tensors parametric over size in certain contexts, such as the semantics of hypergraphs where edges may have any number of inputs or outputs.
Tensors with fixed number of indices can be naturally lifted to dimensionless tensors by taking them to be zero on improperly-sized arguments, so we do not lose any expressiveness.

In practice we require the base type $R$ to have an associated semiring structure, comprising associative, commutative, and distributive addition and multiplication operations, all with respect to some equivalence relation on $R$. 
The type $A$ is generally required to be \emph{summable}, meaning there is a designated list of elements of $A$ which can be summed over to perform a contraction.
When these structures are present on $R$ and $A$, we can define the product and contraction of tensors as usual.

\subsection{Hypergraphs}

To implement hypergraphs, we have to make some modifications to the usual mathematical definition to get things to work well in a proof assistant. 
The standard definition involves a set of edges and a set of vertices, and requires that edges reference only vertices in that designated set. 
In a proof assistant, sets of general objects are not well-behaved (without assuming axioms, it is necessary to consider sets up to some equivalence relation). 
Moreover, requiring edges' input and output vertices to be contained in a given set would entail carrying around a proof of that fact, either within the definition itself or as an additional predicate that the graph is well-formed. 
Instead, we can simply ensure that this condition is met by 
computing the set of vertices incident to any edge (including input and output edges for hypergraphs with interfaces).
This alone would be insufficient as hypergraphs can have isolated vertices. We store an additional field of \coqe{hypervertices} which is used to record extra vertices that may not be mentioned in the edges. 
Finally, we replace the set of edges with a finitely-supported map from positives to edges, as this makes it easier to reason about functions that combine multiple hypergraphs.
Edges have labels of a type \coqe{T} that is equipped with some equivalence relation, using the Rocq-std++ library's~\cite{stdpp} \coqe{Equiv} typeclass, which we will here denote by $\equiv_T$. 
This equivalence relation makes it possible to reason about more general theories in which generators may be equivalent without being literally equal (for example, if the generators are parameterized by Rocq's rational numbers).
We present our definition in \Cref{lst:hypergraph}.

\begin{lstlisting}[label=lst:hypergraph, caption=Hypergraph built with \coqe{Pset} and \coqe{Pmap} from stdpp]
Record HyperGraph {T} := mk_hg {
  (* The edges of the hypergraph *)
  hyperedges : Pmap (T * list positive * list positive);
  (* Additional vertices of the hypergraph, which are often
    disjoint from the referenced vertices of [hyperedges]
    (in practice, we only care about the subset of [hypervertices]
    not referenced in [hyperedges], but do not enforce disjointness) *)
  hypervertices : Pset;
}.
\end{lstlisting}

We extend this notion to a notion of hypergraphs with interfaces.
To more closely match the primary work behind this rewrite engine~\cite{bonchi2}, we use the term cospan in our library instead of interface.
In order to avoid bringing in unnecessary category theory definitions, we opted to use the term interfaces in this paper.
It should be noted that they are the same in practice; cospans just serve as a formalization of the notion of interfaces~\cite{bonchi2}.

\begin{lstlisting}[label=lst:cospanhypergraph, caption=Hypergraph with an interface (cospan) with \coqe{n} inputs and \coqe{m} outputs]
Record CospanHyperGraph {T : Type} {n m : nat} := mk_cohg {
  hedges : HyperGraph T;
  inputs : vec positive n;
  outputs : vec positive m;
}.
\end{lstlisting}

We implement our hypergraphs as a record containing the hyperedges and hypervertices of the graph, as described above.
We use the implementations of extensional maps and sets of positives from stdpp~\cite{stdpp}, giving us access to the standard operations. 
Hypergraphs with interfaces are defined as hypergraphs along with vectors specifying which vertices are the inputs and outputs.
We use explicit input and output sizes so that Rocq's typechecking can ensure that compositions are properly sized.


\wjbscomment{[[ Hypergraph translation to tensors ]]}
\wjbscomment{
- If we can turn edge labels into tensors, we can turn hypergraphs into tensors by taking products of the tensor corresponding to each edge - then internal vertices tell us what indices to sum over
}
To give tensor semantics to a hypergraph with interfaces, we require that the type \coqe{T} be interpretable as dimensionless tensors, as recorded by the \coqe{TensorLike} typeclass. 
This records that each element $t$ of $T$ corresponds to a tensor, and in such a way that equivalent elements of $T$ (by $\equiv_T$) are mapped to equivalent (dimensionless) tensors.
We use this representation rather than taking the edge labels to be \coqe{DimensionlessTensor}s so that our algorithms can compare edge labels effectively, as case analysis does not work on \coqe{DimensionlessTensor}s.
This is essential for reflective isomorphism testing, which underpins our automation.

\begin{lstlisting}[label=lst:tensorlike, caption=TensorLike typeclass modified from the original for readability]
Class TensorLike (R A T : Type) {Equiv T} := {
  interpretTensor (t : T) : DimensionlessTensor R A;
  interpretTensorProper :: Proper (equiv ==> equiv) interpretTensor
}.
\end{lstlisting}

These tensors give semantic interpretation to hypergraphs against which we can verify the correctness of our operations and relations.
Equivalence of tensor semantics is the ultimate notion of hypergraph equivalence that we consider, but the power of hypergraphs is that their \textit{syntactic} equivalence, isomorphism, can be effectively computed without explicitly translating them to tensors.
With our modified definition of hypergraphs with interfaces, the definition of isomorphism is slightly modified to accommodate the small difference in the edge map, vertex set, and relation on labels. 
Specifically, the injective functions $f_v$ and $f_e$ witness an isomorphism between the hypergraphs with interfaces $G$ and $H$ if:
\begin{enumerate}
    \item[a.] The interface (inputs and outputs) of $H$ is the application of $f_v$ to the interface of $G$
    \item[b.] There is an edge $(t,v,w)$ indexed by $k$ in $G$ if and only if there is an edge $(t',v',w')$ indexed by $f_e(k)$ such that $t\equiv_T t'$, $v' = f_v(v)$, and $w'=f_v(w)$
    \item[c.] The (computed) vertex set \coqe{vertices H} is the image of \coqe{vertices G} by $f_v$.
\end{enumerate}
We prove the induced relation of isomorphism is an equivalence relation on hypergraphs with interfaces, and moreover that isomorphic graphs have equivalent tensor semantics. 
This means semantic equalities can be proven by checking the syntactic condition of hypergraph 
isomorphism.



\wjbscomment{[[ Hypergraph isomorphism testing ]]}

Because we have defined hypergraphs computationally, we can implement a reflective partial decision procedure for isomorphism checking. 
This takes the form of an algorithm for finding isomorphisms along with a proof that this algorithm is correct, meaning that isomorphism checking reduces to simple computation with Rocq's reduction engines~\cite{vmcompute}.
Graph isomorphism can be computationally expensive to compute, but in this context it proves tractable. 
Because the hypergraphs have directed, labeled hyperedges, as well as interfaces, there is very little ambiguity as to what map could induce an isomorphism. 
In practice, we find that the graph isomorphism test scales well to moderately sized goals, but can become slow on very large graphs.
Empirically, it runs in less than a second on all the non-synthetic goals we encountered in our examples, which had sizes up to several dozen edges.

For an SMC term built up from generators by composition and stacking, we can represent the term by a hypergraph with interface, assuming we can represent the generators.
To use this representation in a proof assistant, we must verify the correctness of this translation; in \tensrocq this is with respect to the underlying tensor semantics.
For an SMC with semantics in tensors over some semiring, we can express that a hypergraph with interface represents a given SMC term by saying they have the same semantics as tensors.
This allows us to prove results in existing developments that admit tensor semantics, rather than just abstractly reasoning about graphs.

\subsection{PROP implementation}

To reason about SMC terms via computational reflection, we need a syntactic representation for those terms.
For simplicity, we consider single-sorted tensors, where each index has the same dimension. Hence, the only relevant type data is the number of inputs and outputs. For that reason, we restrict to SMCs whose objects are natural numbers, which in the literature are called PROduct categories with Permutations, or PROPs~\cite{maclane1965categorical}.
We will furthermore assume the existence of cap ($\cap$) and cup ($\cup$) constructors, which are present whenever a PROP is \textit{compact closed}. These allow us to reason about connectivity in a wider variety of domains. 
For instance, this means the equations given in \cref{dfn:autonomous} reduce to isomorphisms of hypergraphs, and therefore do not need to be explicitly encoded.
Moreover, cap and cup can be implemented as tensors, so we do not need an explicit cap and cup in the SMC under consideration.
Our SMC rewriting tactic does not produce caps and cups, so it can be used in an SMC without them. 


\begin{lstlisting}
Inductive APROP {T : Type} : nat -> nat -> Type :=
  | Aid n : APROP n n
  | Aswap n m : APROP (n + m) (m + n)
  | Acup n : APROP 0 (n + n)
  | Acap n : APROP (n + n) 0
  | Acompose {n m o} (ap1 : APROP n m) (ap2 : APROP m o) : APROP n o
  | Astack {n1 m1 n2 m2} (ap1 : APROP n1 m1) (ap2 : APROP n2 m2) :
    APROP (n1 + n2) (m1 + m2)
  | Agen (t : T) n m : APROP n m. (* Generators of the category, encoded using T *)
\end{lstlisting}

APROPs have natural semantics as tensors, assuming a \coqe{TensorLike} instance is present for \coqe{T}.
They can also be converted to hypergraphs with interfaces by direct computation, so we can use the fast isomorphism testing function on hypergraphs to prove equivalences of APROPs simply by evaluating the test on the hypergraph semantics of the APROP terms.

\begin{lstlisting}
Fixpoint APROP_graph_semantics {T n m} (ap : APROP T n m) : CospanHyperGraph T n m :=
  match ap with
  | Aid n => id_graph n
  | Aswap n m => swap_graph n m
  | Acup n => cup_graph n
  | Acap n => cap_graph n
  | Acompose ap1 ap2 =>
      compose_graphs (APROP_graph_semantics ap1) (APROP_graph_semantics ap2)
  | Astack ap1 ap2 =>
      stack_graphs (APROP_graph_semantics ap1) (APROP_graph_semantics ap2)
  | Agen t n m => graph_of_tensor t n m
  end.
\end{lstlisting}

To perform hypergraph rewriting on APROP terms, we need not only to convert APROPs to hypergraphs, but also to convert hypergraphs back to APROPs.
However, not every hypergraph represents a valid APROP term: for example, a non-monogamous hypergraph cannot be represented by an SMC~\cite{bonchi2}.
As such, a function from hypergraphs to APROP terms is necessary partial.
We implement such a partial function which attempts to convert a hypergraph into an SMC term by proceeding "left-to-right", i.e. from inputs to outputs, extracting hyperedges by horizontal composition. 
In the case of terms arising from an SMC, the resulting hypergraph is acyclic monogamous, and indeed any acyclic monogamous hypergraph can be represented as an SMC term, hence an APROP term.
Because double-pushout rewriting preserves acyclic monogamous hypergraphs, every hypergraph our tactic generates will be representable as an APROP term~\cite{bonchi2}.

\subsection{Rewriting}

We implement rewriting modulo SMC structure based on the double-pushout rewriting described in \cite{bonchi2}. 
In broad strokes, given a term $H$ and some known equivalence $L \equiv R$, to rewrite $L$ into $R$ within $H$ we must decompose $H$ into a form in which $L$ appears directly as a subterm. 
For acyclic monogamous cospans of hypergraphs, this decomposition (the pushout complement) can be computed by separating $H$ into edges with paths into $L$ and edges with paths out of $L$ (assuming $L$ is convex). 

We have several tactics for rewriting, depending on the specific context in which the rewriting occurs, but they have the same structure.
For simplicity, we will focus on the tactic which operates on APROPs directly; the tactics which apply to other theories differ only in adding additional steps on either end to first convert the goal into a statement about APROPs and at the end convert the result back to the target theory.
As such, the illustrative example is to suppose we have APROP terms $H$, $L$, and $R$, and there is a proof of $L \equiv R$ relative to the APROP tensor semantics.
First, we convert the target term $H$ and both sides of the equivalence $L \equiv R$ to hypergraphs $G_H$, $G_L$, and $G_R$, by direct computation.
We call an unverified matching function to try to find a subgraph of $G_H$ isomorphic to $G_L$; this function is a slight modification of the isomorphism test to remove the condition the mapping be surjective. 
Given a match is found, we call the \coqe{decompose} function to split $G_H$ into a composition  $G_{C_1} ; I \otimes G_{L'} ; G_{C_2}$. 
If matching and decomposition were successful, the hypergraph $G_{L'}$ should be isomorphic to $G_L$, and thus $G_H$ should be isomorphic to $G_{C_1} ; I \otimes G_{L} ; G_{C_2}$.
The tactic then attempts to convert $G_{C_1}$ and $G_{C_2}$ to APROP terms $C_1$ and $C_2$, which will always be possible when the original terms are SMC terms.
Therefore, we should have $H \equiv C_1 ; I \otimes L ; C_2$, which the tactic proves using an isomorphism test. 
Finally, we can replace $L$ with $R$ because composition and stacking preserve tensor semantics.

The correctness of this rewrite depends only on isomorphism testing and composition, with no assumption that the matching, decomposition, and most importantly term extraction from hypergraphs to APROPs are correct.
These functions are would be difficult to verify, and more importantly, omitting a verification requirement allows us to optimize performance and result quality.
This creates a separation of concerns: the rewrite engine can be implemented and improved quickly, with isomorphism testing ensuring correctness throughout.
However, requiring this additional check does incur a performance penalty.
Experimentally, the isomorphism check appears to take up to a third of the time of a rewrite, especially on larger goals. 
If matching, decomposition, and term extraction were verified, that check could be avoided, speeding up the tactic.

Allowing unverified rewriting engines also expands the scope of applications.
As we assume nothing about the matching algorithm, another tool could perform the matching, or even selection of lemma to rewrite, with the assurance that any rewrite performed is correct.

\section{Signature Reasoning}

There are two ways to use \tensrocq, \textit{theories} and \textit{instantiations}. Building a theory allows you to give a collection of generators and rewrite rules as an APROP which can then be used to build terms and perform rewriting. We call this axiomatized reasoning since it relates to small set of axioms rather than the underlying structures (this is not to be confused with axiomatization in Rocq). Instead, we define a structure called the Signature which allows the user to provide generators, an equivalence relation, and rewrite rules. 

\begin{lstlisting}
Structure Signature {A : Type} {SA : Summable A, EqA : EqDecision A} := {
  #[canonical=yes] gens : Type;
  #[canonical=no]  gens\_equiv :: Equiv gens;
  #[canonical=no]  gens_equivalence :: @Equivalence gens equiv;
  #[canonical=no]  rules : forall {n m}, relation (APROP gens n m);
}.
\end{lstlisting}

We reason about these rewrites using the same tensor semantics as the rest of the paper. However, we are not given a ring by the user. Instead we generate an appropriate polynomial ring to use our rewriting system. None of this is exposed to the user, and they are able to define a Signature and use our rewrite system immediately. To show this, we will take a look at a simple Frobenius algebra. A Frobenius algebra is an example of a PROP given by four generators
\[
m : 2 \to 1, \quad 
u : 0 \to 1, \quad
n : 1 \to 2, \quad
v : 1 \to 0.
\]
satisfying a small collection of rewrite rules given by
\[
    m * id ; m = id * m ; m, \quad
    u * id ; m = id, \quad
    id * u ; m = id,
\]\[
    n ; n * id = n ; id * n, \quad 
    n ; v * id = id, \quad
    n ; id * v = id,
\]\[
    n * id ; id * m = id * n ; m * id.
\]
The first three rules show $m$ and $u$ form a monoid, the second three show $n$ and $v$ form a comonoid. The final rule shows how the monoid given by $m$ and $u$ interacts with the monoid given by $n$ and $v$. There are two equivalent formulations of this condition.
\[
  n * id ; id * m = m ; n \quad
  id * n ; m * id = m ; n
\]
If $m$ and $n$ are both commutative, then these rules imply our previous rule. However if we have a non-commutative Frobenius algebra we end up needing both. We will now show how to define this Frobenius algebra in \tensrocq and derive the two other forms of the Frobenius condition.

First, we must define a relation with respect to which we will give our rules. We give one that simply states our terms share the same dimension, but any relation would do. We also provide notations for the generators beforehand to make writing the signature simpler.

\begin{lstlisting}
Notation "x ≡ y" :=
  (existT _ (existT _ (x%APROP, y%APROP)) :
    {n & {m & (APROP (fin 4) n m * APROP (fin 4) n m)%type}})
  (at level 70).

Notation m := (Agen (0%fin) 2 1).
Notation u := (Agen (1%fin) 0 1).
Notation n := (Agen (2%fin) 1 2).
Notation v := (Agen (3%fin) 1 0).

Definition Frob : Signature bool := {|
  gens := fin 4;
  gens_equiv := eq;
  rules := rules_of_rule_list [

    (* (m, u) forms a monoid *)
    (* rule assoc : *) m * Aid 1 ;' m ≡ Aid 1 * m ;' m ;
    (* rule unitL : *) u * Aid 1 ;' m ≡ Aid 1 ;
    (* rule unitR : *) Aid 1 * u ;' m ≡ Aid 1 ;
    
    (* (n, v) forms a comonoid *)
    (* rule coassoc : *) n ;' n * Aid 1 ≡ n ;' Aid 1 * n ;
    (* rule counitL : *) n ;' v * Aid 1 ≡ Aid 1 ;
    (* rule counitR : *) n ;' Aid 1 * v ≡ Aid 1 ;
    
    (* rule frob : *) n * Aid 1 ;' Aid 1 * m ≡ Aid 1 * n ;' m * Aid 1

  ];
|}.
\end{lstlisting}

Once we have defined our signature, we must lift these notions from the relation $\equiv$ specifying the rules of the signature to a relation $==$ expressing that terms have the same semantic interpretation in the generated semiring. 
This comprises restating each rules with this new relation, following a boilerplate procedure.

\begin{lstlisting}
Notation "x == y" :=
  (APROP\_semantics x ≡ APROP\_semantics y) (* Semantics into the generated semiring *)
  (at level 70).

(* (m, u) forms a monoid *)
Lemma assoc : m * id ;' m == id * m ;' m.
Proof. apply rules_hold. repeat constructor. Qed.
Lemma unitL : u * id ;' m == id.
Proof. apply rules_hold. repeat constructor. Qed.
Lemma unitR : id * u ;' m == id.
Proof. apply rules_hold. repeat constructor. Qed.

(* (n, v) forms a comonoid *)
Lemma coassoc : n ;' n * id == n ;' id * n.
Proof. apply rules_hold. repeat constructor. Qed.
Lemma counitL : n ;' v * id == id.
Proof. apply rules_hold. repeat constructor. Qed.
Lemma counitR : n ;' id * v == id.
Proof. apply rules_hold. repeat constructor. Qed.

Lemma frob : n * id ;' id * m == id * n ;' m * id.
Proof. apply rules_hold. repeat constructor. Qed.
\end{lstlisting}


This is also the point where we give names to our rules, as we are now lifting them to be used in Rocq proofs. The \coqe{rules_hold} lemma shows our equality holds for the generated semiring if that rule is present in the list given as a part of our APROP. Now that our rules are encoded, we can use the string rewriting tactic \coqe{srw} to perform string rewrites over Frobenius terms. Practically this means for Frobenius terms, we no longer have to care about associativity of terms but instead are able to reason as if we were looking at a string diagram. We now prove the first of our other two Frobenius rules.

\begin{lstlisting}
Lemma frobL : n * id ;' id * m == m ;' n.
Proof.
  transitivity (u * n * id ;' m * m)%APROP; 
  [srw unitL; smcat|].
  (* u * n * [[ id ]];' m * m == m;' n *)
  srw <- frob.
  (* u * [[ id 2 ]];' Aswap 2 1;' [[ id ]] * (n * [[ id ]];' [[ id ]] * m);' (Aswap 2 1;' m * [[ id ]];' sw) == m;' n *)
  srw assoc.
  (* u * [[ id 2 ]];' n * [[ id 2 ]];' [[ id ]] * ([[ id ]] * m;' m) == m;' n *)
  srw frob.
  (* m * u;' sw;' ([[ id ]] * n;' m * [[ id ]]) == m;' n *)
  srw unitL.
  (* m;' (n;' sw);' sw == m;' n *)
  smcat.
Qed.
\end{lstlisting}

Two tactics are present here, \coqe{smcat} which solves simple equations using hypergraph isomorphism. This tactic is sufficient to solve equations where the only difference is syntactic, up to connectivity. If the underlying hypergraphs are isomorphic, then \coqe{smcat} can solve the equation. The second tactic \coqe{srw} is the true rewriting tactic. It takes a term \coqe{lhs == rhs} and attempts to perform a replacement of the \coqe{lhs} term with the \coqe{rhs} term in the current goal as a hypergraph.

We can now use \coqe{frobL} in further rewrites. This means we can prove the other Frobenius rule with just two rewrites.

\begin{lstlisting}
Lemma frobR : id * n ;' m * id == m ;' n.
Proof.
  srw <- frob.
  (* n * [[ id ]];' [[ id ]] * m == m;' n *)
  srw frobL.
  (* m;' n == m;' n *)
  smcat.
Qed.
\end{lstlisting}

This part of \tensrocq is implemented in such a way to give us feature parity with Chyp. 
Much effort has been put towards making this as simple as possible, including automatically generating \coqe{APROP_semantics} and lemmas such as \coqe{rules_hold} to allow for simple boilerplate code in order to lift lemmas from the rules of the signature to Rocq's rewrite system.
Nevertheless, it is still somewhat inconvenient to include this boilerplate, and it is possible that it could be avoided with the development of a Rocq plugin.
We also currently require that all generators and relations of the signature are given at once, whereas Chyp allows them to be added throughout the development.
This difference arises primarily because we give semantic meaning to our rewrite system from the signature itself, so modifying the signature changes the semantics.

Our semantic backing allows us to prove our syntactic reasoning sound: Any semantics for APROPs over the signature's generators that satisfies the relations of the signature will also satisfy any relation we prove syntactically. 
This means that a theory can be developed once and applied to many different projects. 
However, this method of reasoning still requires using signatures, which is not quite suitable for reasoning about existing developments with their own notions of semantic correctness.
In the next section, we will discuss how we bridge the gap between our APROPs and  existing work, showing that \tensrocq is a tool for reasoning not just about the theory of SMCs, but also concrete instances of SMCs, using powerful diagrammatic rewriting.

\section{Instantiated Reasoning: The ZX Calculus}

A major advantage of working in a proof assistant is that \tensrocq can make statements grounded in mathematical descriptions of semantics, not just abstract rewrite systems. 
To apply \tensrocq reasoning to other projects, we provide a framework to apply its tactics to any goal that can be framed using tensor semantics.

As our rewrite engine is based on the APROP structure, it is necessary to frame goals in that context. 
However, a vast number of theories can be described as symmetric monoidal categories, and we do not wish to constrain projects to using the APROP type directly\footnote{For instance, the theory of rings can be described as an SMC with generators for $0$, $1$, $+$, and $\ast$ and SMC equations representing the axioms. 
}
Therefore, we provide a system of typeclasses that can be instantiated to enable \tensrocq to reason about an existing project.

The general process by which our rewrite engine can be used on a target SMC is as follows:
\begin{enumerate}
    \item Specify an APROP representing the theory by defining generators and their semantics.
    \item Define tensor semantics for the target theory, and show that terms with equal tensor semantics are equivalent in the theory.
    \item Provide the definitions of sequential and parallel composition in the target theory via \coqe{APROPlike}, and prove their tensor semantics are the expected contraction and product.
    \item Instantiate typeclasses describing how to convert a term in the target theory to an APROP with equivalent tensor semantics (which we call ``quoting'' the term), and how to convert an APROP to a term in the target theory (which we call ``denoting'' the APROP).
    \item Instantiate our rewrite tactics with these interfaces. 
\end{enumerate}

This system is designed to make minimal expectations of the target theory.
We do not require proofs that every term is represented by an APROP, nor that every APROP represents a term; the tactics will only work in contexts where the relevant terms can be represented by APROPs.
Providing quotation and denotation by user-instantiated typeclasses means tactics can be used in contexts with derived terms that are not written directly as generators, increasing their flexibility. 
These quotation and denotation processes can also be extended with new terms throughout the development of the project without needing to redefine the APROP interpretation.

For a practical example, we describe the application of \tensrocq to an existing project named \vyzx~\cite{lehmann2026vyzxformalverificationgraphical}. 
\vyzx is a formally verified implementation of \zxcalc, a symmetric monoidal category and diagrammatic calculus used for quantum computing.
To use our rewrite engine, we first have to have an idea of the theory we are targeting. In the case of \vyzx, \zxdiags are built as an inductive datastructure.

\begin{lstlisting}[caption=CoreData/ZXCore.v in the VyZX repository]
Inductive ZX : nat -> nat -> Type :=
  | Empty : ZX 0 0
  | Cup  : ZX 0 2
  | Cap  : ZX 2 0
  | Swap : ZX 2 2
  | Wire : ZX 1 1
  | Box  : ZX 1 1
  | X_Spider n m (α : R) : ZX n m
  | Z_Spider n m (α : R) : ZX n m
  | Stack {n_0 m_0 n_1 m_1} (zx0 : ZX n_0 m_0) (zx1 : ZX n_1 m_1) : 
          ZX (n_0 + n_1) (m_0 + m_1)
  | Compose {n m o} (zx0 : ZX n m) (zx1 : ZX m o) : ZX n o.
\end{lstlisting}

Equivalence on \vyzx diagrams is typically given as \textit{proportionality} of semantics by a constant, nonzero complex multiple \coqe{c : C}, denoted in \vyzx by the relation $\propto\![c]$.
However, the definition of equivalence corresponding directly to the equality of tensor semantics is $\propto=$, which stands for $\propto\![1]$, i.e. semantic equality.
For this instantiation, we chose to only work with $\propto=$, as \vyzx provides tools to automatically convert statements of proportionality to statements of semantic equality by adding a scalor factor encoded as a \zxdiag. 
To represent \zxdiags, our APROP type is chosen to be \coqe{option (bool * R + C)}.
The value \coqe{None} stands for the Hadamard box, \coqe{Box}, while the \coqe{C} represents a scalar factor via \coqe{zx_of_const}. The values \coqe{(false,\alpha)} and \coqe{(true,\alpha)} represent the \coqe{Z_Spider} and \coqe{X_Spider}, respectively, with phase \coqe{\alpha}.
These generators are sufficient to represent all other definitions of \vyzx. 
We then provide tensor semantics for this APROP by giving the standard tensor interpretations of each generator, completing step (1).

The semantics of \vyzx diagrams are given as \qlib matrices~\cite{QuantumLib}, which can be easily converted to and from tensors.
Therefore, we can easily show that \zxdiags with equivalent tensor semantics have equivalent matrix semantics, which is step (2).
Providing the composition and stacking of \coqe{ZX} completes step (3).

\begin{lstlisting}[caption=examples/VyZXCRewriting.v, label=lst:zxaproplike]
#[refine] Instance ZX_APROPlike : APROPlike C bool ZX (@Compose) (@Stack) := {
  interpretDiagram n m zx := ZX_tensor_semantics zx;
}.
Proof.
  abstract (intros n m d d' Heq%matrix_of_tensor_of_equiv;
    rewrite 2 ZX_tensor_semantics_correct in Heq;
    prep_matrix_equivalence;
    exact Heq).
  abstract (easy).
  abstract (easy).
Defined.  
\end{lstlisting}

Once we have built an APROPLike for our target theory, we must be able to quote between our APROP and the target theory to allow translation.
This is accomplished through mostly boilerplate code using the DiagramQuote typeclass. 
The statement 
\begin{coq}
DiagramQuote (APROPlikeD:=APROPModel_for_theory) (TensT:=APROP_tensor_interpretation) d a.    
\end{coq}
says that the diagram \coqe{d} in our target theory is semantically equivalent to the term \coqe{a} in our APROP.

\begin{lstlisting}
Local Notation Quote := (DiagramQuote (APROPlikeD:=ZX_APROPlike) (TensT:=ZXCCALC)).
\end{lstlisting}

For example, we can take a term \coqe{n_wire n} and the term \coqe{Aid n} and show they are semantically equivalent with the instance \coqe{zx_quote_n_wire_n}.
\begin{lstlisting}
#[export] Instance zx_quote_n_wire n : Quote (n_wire n) (Aid n).
Proof.
  constructor.
  apply matrix_of_tensor_inj.
  rewrite ZX_tensor_semantics_correct.
  rewrite matrix_of_tensor_delta.
  now rewrite n_wire_semantics.
Qed.
\end{lstlisting}

Implementing a quotation instance for all of the basic structures allows us to use typeclass resolution to translate terms in our target theory into APROPs.
Other terms can be given quotations later as desired, if new derived constants are defined.
To make typeclass resolution behave well\footnote{Specifically, we want \texttt{DiagramQuote} to examine the concrete term and try to generate an APROP, while \texttt{DiagramDenote} should examine the APROP to generate a concrete term. Combining the two can lead to major performance issues, and gives less control over the generated terms.}, we must separately specify the other conversion, from APROP to \zxdiag.
For this we use the typeclass \coqe{DiagramDenote}, which follows the same pattern of boilerplate.
Note that this typeclass has precisely the same meaning as \coqe{DiagramQuote}, so when a term is quoted and denoted identically, the proofs of the corresponding instances can be repeated verbatim.

\begin{lstlisting}
Local Notation Quote := (DiagramDenote (APROPlikeD:=ZX_APROPlike) (TensT:=ZXCCALC)).
\end{lstlisting}

\begin{lstlisting}
#[export] Instance zx_denote_n_wire n : Quote (n_wire n) (Aid n).
Proof.
  constructor.
  apply matrix_of_tensor_inj.
  rewrite ZX_tensor_semantics_correct.
  rewrite matrix_of_tensor_delta.
  now rewrite n_wire_semantics.
Qed.
\end{lstlisting}

Once the quoting and denoting of step (4) are finished, we can now translate automatically between APROP terms and terms of our target theory.
The last step needed is to instantiate new rewrite tactics for this target theory which use these instances to move freely between our APROPs and the target terms, so that the tactics can determine the proper instances to apply.

\begin{lstlisting}[caption=Rewriting tactics lifted to the VyZX grammar, label=lst:rewrites]
(* Prove that ZX terms corresponding to isomorphic hypergraphs are \propto= *)
Ltac zxcat := wild_cat ZX_APROPlike ZXCCALC.

(* Simplify the LHS by removing extraneous identities *)
Ltac zxclean_lhs := wild_clean_lhs ZX_APROPlike ZXCCALC.

(* Simplify the RHS by removing extraneous identities *)
Ltac zxclean_rhs := wild_clean_rhs ZX_APROPlike ZXCCALC.

(* Simplify both sides of the goal by removing extraneous identities *)
Ltac zxclean := wild_clean ZX_APROPlike ZXCCALC.

(* Rewrite [lem] at occurence number [match_num] in the LHS, up to SMC equivalence *)
Ltac zxrw_lhs lem match_num := wild_rw_lhs ZX_APROPlike ZXCCALC lem match_num.

(* Rewrite [lem] at occurence number [match_num] in the RHS, up to SMC equivalence *)
Ltac zxrw_rhs lem match_num := wild_rw_rhs ZX_APROPlike ZXCCALC lem match_num.

(* Rewrite [lem] at occurence number [match_num] in the goal, up to SMC equivalence *)
Ltac zxrw lem match_num := wild_rw ZX_APROPlike ZXCCALC lem match_num.
\end{lstlisting}

With that, \coqe{zxrw} can now be used in proofs in our target theory. To show the power of this rewriting tactic, we take a look at an example proof from \vyzx,  which shows that three CNOT gates implement a swap operation: \coqe{Lemma _3_cnot_swap_is_swap : _CNOT_ ⟷ _NOTC_ ⟷ _CNOT_ ∝[/ (2 * √2)] ⨉.}~\cite{lehmann2026vyzxformalverificationgraphical}. In \vyzx, this proof took 45 lines of code.
A large majority of these lines are dedicated to associating terms to expose a subdiagram exactly, so that Rocq's \coqe{setoid_rewrite} tactic can perform the rewrite.
This reassociation code can be annoying to write, and it is very fragile, as it depends intimately on the term structure of the goal.
Using our rewrite tactic, this is reduced to only 17 lines of proof, 5 of which have to do with the specifics of how we integrate with \vyzx's $\propto=$ relation. 
Moreover, our rewrites depend only on the hypergraph interpretations, so are much more resilient to any changes in definition.

We first show that two consecutive \coqe{_CNOT_} gates (in \vyzx, \coqe{Z 1 2 0 * id ; id * X 2 1 0}) cancel:

\begin{lstlisting}
Lemma _CNOT_CNOT_cancels : _CNOT_ ⟷ _CNOT_ ∝[/ 2] n_wire 2.
Proof.
  apply prop_by_iff_zx_scale. 
                      (* Change relation to ∝= by adding in scalar ZX diagram *)
  split. 2:{ apply nonzero_div_nonzero; nonzero. }                                     
                                        (* Show above operation is reasonable *)
  zxrw (@dominated_Z_spider_fusion_top_left 2 0 1 1 0 0). 
                                                    (* Fuse the top Z spiders *)
  zxrw (@dominated_X_spider_fusion_bot_right 2 0 1 1 0 0). 
                                                 (* Fuse the bottom X spiders *)
  rewrite Rplus_0_l.                                     (* Correct rotations *)
  zxrw (to_gadget hopf_rule_Z_X_vert 1 1 1 1 0 0 eq_refl). 
                       (* Remove 2 connections between top and bottom spiders *)
  rewrite Z_is_wire, X_0_is_wire.                       (* Now they are wires *)
  zxcat.                   (* Equivalent because connections are all the same *)
Qed.
\end{lstlisting}

And then we can show that three alternating CNOTs form an identity:

\begin{lstlisting}
Lemma _3_cnot_swap_is_swap : _3_CNOT_SWAP_ ∝[/ (2 * √2)] ⨉.
                           (* _CNOT_ ; _NOTC_ ; _CNOT_ = sw *)
Proof.
  apply prop_by_iff_zx_scale. 
                        (* Change relation to ∝= by adding scalar ZX diagrams *)
  split. 2:{ apply nonzero_div_nonzero, Cmult_neq_0; nonzero. }                                 
                                    (* Show the above operation is reasonable *)
  rewrite cnot_is_swapp_notc at 2.     (* Add in a swap and change the third  *)
  rewrite notc_is_notc_r.                     (* Move the cnot back and forth *)
  zxrw (to_gadget bi_algebra_rule_X_over_Z). 
                        (* Find and automatically perform a bialgebra rewrite *)
  zxrw (to_gadget _CNOT_CNOT_cancels).    (* Cancel the two CNOTs that result *)
  zxrw (symmetry (zx_of_const_mult (/ C2) (/ √ 2))).      (* Fix up constants *)
  rewrite Cinv_mult_distr.                       (* More constant corrections *)
  zxcat.                                             (* Graphs are equivalent *)
Qed.
\end{lstlisting}

These proofs correspond more closely to the diagrammatic version of these proofs than the original \vyzx proofs, as they emphasize diagrammatic transformations and suppress term-level associativity concerns.
We did find some limitations with using our rewrite engine on \vyzx, namely to deal with caps and cups. 
Because we chose to implement caps and cups within the APROP structure, we are able to reason about isomorphism up to the yanking equations. 
However, the matching and decomposition functions currently target SMC terms exclusively, and do not support caps and cups.
Moreover, the ``spiders'' that serve as generators for the \zxcalc have very strong permutative properties, with the ability to convert between inputs and outputs, as well as absorb permutations into the generators. 
\tensrocq makes reasoning about these permutations easier, as yanking can be inferred, but it is still necessary to manually state the transformations to apply to spiders to make them isomorphic as hypergraphs.
Encoding this permutativity directly could allow for much stronger isomorphism checking, or even rewriting, up to the full notion of only connectivity matters for the \zxcalc.

\section{Future Work}

The primary gap that remains between this tool and existing work in the area is visualization. Chyp~\cite{chyp} has an integrated visualizer which allows for the easy identification of which diagrammatic rules to apply. As TensorRocq is embedded within Rocq, visualization and interactivity are more difficult to implement. Existing work from the \vyzx project~\cite{lehmann2026vyzxformalverificationgraphical} shows how we can interact with the language server protocol (LSP) to generate visualizations, but it remains to be seen how interactivity could work with our embedded rewrite engine and language. Recent work by Damien Pous~\cite{pous2026string} shows that it is possible to copy Rocq code to an external visualizer, perform graphical rewrites there, and extract the proofs back to Rocq. Integrating this with Rocq-LSP would give the best of both worlds in terms of allowing easy visualization and rewriting. 

Further extensions and improvements of the rewrite engine, including a full verification of the system remain open problems. The theoretical backing of the engine is strong and we have found it can be applied to existing projects such as \vyzx. Concrete diagrams where we have fixed dimension arguments are sufficient for many applications, including circuits where the number of inputs or outputs to a process is fixed.
However, adding the ability to reason over generators with a variable number of inputs or outputs would be a major improvement. This is feasible with the use of sized hypergraphs, hypergraphs with a notion of ``size'' for each vertex, with each vertex essentially representing a bundle of hypergraph vertices. These sizes can then be represented syntactically, allowing for the underlying hypergraph to be of finite size and therefore computable.

Another extension is the ability to match on parametric generators, such as the phases of ZX spiders, possibly up to some equivalence relation. This comprises extending the matching algorithm to track the assignment of parameters for each match and incorporating this within the tactic infrastructure to automatically determine the parameters associated with a given rewrite rule.

The checked-rewrite structure of our rewrite engine also makes it much easier to extend to categories with additional properties. For instance, currently we can only reason about caps and cups for the purposes of isomorphism, not rewriting, as the matching engine is based on decomposition for a symmetric monoidal category. Implementing matching and decomposition procedures would make the rewrite engine significantly more powerful in domains with caps and cups. In principle, it should also be possible to encode other structural conditions, such as the permutativity of the Z and X spiders in the \zxcalc, to enable rewriting up to these more flexible notions of connectivity.

Further formalisation of more general PROPs as abstract tensor systems~\cite{Kissinger2014} could allow one to reason about a broader class of theories than those with a concrete tensor semantics. However, for PROPs that do not form hypergraph categories, some care must be taken to maintain well-formedness properties such as monogamy, acyclicity, and convexity~\cite{bonchi1} at the level of abstract tensors.


We also hope to see this rewrite tactic applied to new domains. The ability to ignore associativity is an invaluable tool for rewriting. There are a several types of project for which \tensrocq is a natural fit. Using \tensrocq within linear algebra libraries such as the ones found in mathematical components~\cite{mahboubi2021mathematical} would be quite straightforward, as not only are matrices are naturally interpreted as tensors, but also as PROPs. We have explored one quantum computing library, \vyzx, but other libraries like Qbricks~\cite{chareton2021qbricks} and SQIR~\cite{hietala2021sqir}, provide different symmetric monoidal structures that could amenable to verification in \tensrocq.
We also hope that the availability of effective tooling for SMCs will encourage the development of more formalizations.

\section{Conclusion}

\tensrocq closes the gap between on-paper proof and formal verification for symmetric monoidal categories.
We enable true diagrammatic reasoning, not only proving structural equivalences automatically, but enabling automatic, verified rewriting modulo SMC equivalence.
This allows SMC terms to be thought of as diagrams, with associativity handled being automatically by powerful, efficient automation. 
We have also demonstrated how our extensible framework can be applied to existing projects, like the \vyzx \zxcalc library~\cite{lehmann2026vyzxformalverificationgraphical}. 

Diagrammatic rewriting makes proofs much shorter and more readable than the previously-necessary manual manipulation of associativity.
\tensrocq's diagrammatic rewriting is verified with respect to tensor semantics, so it can be used soundly in any theory interpretable within tensors.
We also provide a structure to define abstract rewriting systems by generators and relations, in the style of Chyp~\cite{chyp}.
These abstract systems can also be instantiated with concrete implementations, and we prove that equations proved in the abstract system hold in the concrete instantiation.

\begin{acks}
This material is based upon work supported by the Air Force Office of Scientific Research under award numbers FA95502310361 and FA95502310406. This work is also supported by the Engineering and Physical Sciences Research Council grant reference EP/Z002230/1: (De)constructing quantum software (DeQS).
\end{acks}


\bibliography{references}

\end{document}